\documentclass[conference]{IEEEtran}
\usepackage{amsthm,amsmath,amssymb,esint,color,graphicx,overpic,bbm,cite,hyperref,wrapfig,acronym,subcaption,latexsym,paralist,xspace,array,multirow,url}
\usepackage{times}\definecolor{light-gray}{gray}{0.95}
\usepackage[left=2cm, right=2cm, top=1.5cm, bottom=1.5cm,includeheadfoot]{geometry}
\usepackage[linesnumbered]{algorithm2e}
\newtheorem{theorem}{Theorem}
\newtheorem{lemma}{Lemma}

\begin{document}
\title{Computational Optimal Transport for 5G  Massive  C-RAN Device Association}

\author{
Georgios Paschos, Nikolaos Liakopoulos, Merouane Debbah, and Tong Wen \\
\vspace{0.1in} Huawei Technologies
}

\maketitle

\addtolength{\floatsep}{-\baselineskip}
\addtolength{\dblfloatsep}{-\baselineskip}
\addtolength{\intextsep}{-\baselineskip}
\addtolength{\textfloatsep}{-\baselineskip}
\addtolength{\dbltextfloatsep}{-\baselineskip}
\addtolength{\abovedisplayskip}{0ex}
\addtolength{\belowdisplayskip}{0ex}
\addtolength{\abovedisplayshortskip}{0ex}
\addtolength{\belowdisplayshortskip}{0ex}
\setlength{\abovecaptionskip}{0ex}
\setlength{\belowcaptionskip}{0ex}

\begin{abstract}

The massive scale of future wireless networks will create computational bottlenecks in performance optimization. In this paper, we study the problem  of connecting mobile traffic to Cloud RAN (C-RAN) stations. To balance station load, we steer the traffic by designing device association rules. 
The baseline association rule connects each device  to the station with the strongest signal, which does not account for interference or traffic hot spots, and leads to load imbalances and performance deterioration. Instead,  we can formulate an optimization problem to decide centrally the best  association rule at each time instance. However, in practice this optimization has such high dimensions, that even linear programming solvers fail to solve. To address the challenge of massive connectivity, we propose an approach based on  the theory of optimal transport, which studies the economical transfer of probability between two distributions. 
Our proposed methodology can further inspire scalable algorithms for massive optimization problems in wireless networks. 
\end{abstract}

\section{Introduction}

We revisit the problem of \emph{device association}  in the setting of  massive connectivity. 
In this problem, we seek to find a rule to associate mobile traffic to certain serving stations such that the incurred load is balanced across the available stations.
Although a plurality of association methodologies are available in the literature, here we focus  on the underexplored aspect of scalability; we seek load-balancing associations for thousands of devices to hundreds of stations, a setting where  even linear program solvers become cumbersome. 
To address the arising computational challenge, we propose to use \emph{Optimal Transport} (OT) theory, which  studies the  transfer of   masses over a metric space. Recently, an {entropic regularization}  of OT was shown to provide  superfast algorithms for very large OT instances \cite{Cuturi}, making the framework applicable to a wide range of challenging applications, including  image processing and machine learning. \emph{Our goal in this paper is to use regularized OT to derive association rules for a very large number of wireless devices.}

Our centralized approach is motivated by the upcoming 5G wireless networks. 
According to recent reports, telecom operators are increasingly interested in deploying  Cloud-Radio Access Network (Cloud-RAN, or C-RAN) systems, cf.~\cite{report}. 
The architecture of  C-RAN  economizes computation and signal processing by  migrating the computing part of base stations to a central cloud location, and using simple \emph{Remote Radio Heads} (RRH) to broadcast   signals \cite{CRAN}. 
Since all the intelligence is now moved to the C-RAN controller, provisioning connectivity in a large geographical area is centrally decided, motivating the centralized  association of devices to   RRHs. 

A contemporary C-RAN controller manages a big number of RRHs and mobile devices, thus the centralized device association is inherently a large scale problem. 
Furthermore, the number of RRHs and of devices are both expected to increase  in the near future. 
On one hand, the deployment of RRHs will become very dense to improve the effective capacity of the network \cite{network_densification}, while 
on the other hand,  we expect a huge number of  heterogenous smart  IoT devices to connect to 5G mobile networks by 2020--estimated 20.4 billions in \cite{gartner}.  
Since 5G applications can have radically diverse requirements, the traffic footprint of each connected device (IoT or regular) can  vary significantly. \emph{This motivates us to study the large-scale centralized device association  with potentially different traffic requirements per device.}

\section{The Device Association Problem}

\begin{figure}[t!]
  \centering
    \includegraphics[width=0.4\textwidth]{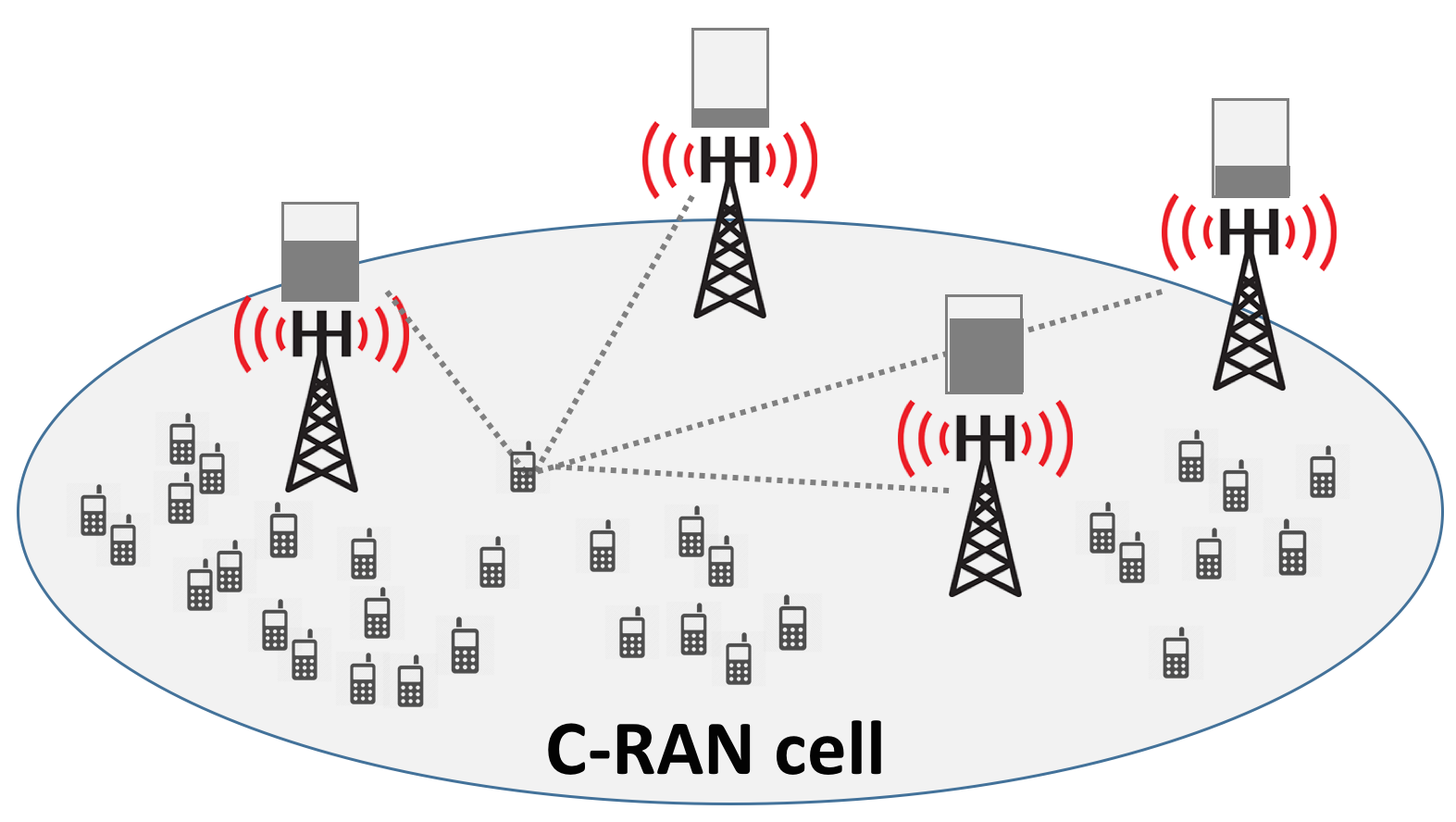}
      \caption{{\small Devices associate to RRHs causing various load levels.}}\vspace{-0.1in}
\end{figure}

\subsection{Downlink Model}
We consider the downlink\footnote{
Our approach applies also to the uplink as long as  the modeled upload rate $R_{ij}$ can be considered independent of the association rule.} transmissions of a large C-RAN cell, containing $\mathcal{I}$ devices and $\mathcal{J}$ RRHs. We call $\pi_{ij}\in [0,1]$ the association variable, which is the fraction of time device $i\in \mathcal{I}$ connects to RRH $j\in \mathcal{J}$. Our objective is to decide  the association rules $\boldsymbol\pi\equiv (\pi_{ij})$ in order to optimize a network performance metric, such as the sum of RRH load or the average job completion time (delay).

In order to evaluate the impact of an  association rule we consider the \emph{download rate} obtained when device $i$ is connected to station $j$ {while receiving exclusive service}; this is  denoted by $R_{ij}$. This rate should be calculated at the granularity of association changes. Typically, we may assume the use of a temporally-fair scheduler  which distributes the station resources to the different users and averages out  fast fading effects.
In this case, a reasonable model for the download rate is:
\[
R_{ij}=W\log(1+\text{SINR}_{ji}),
\]
where, 
 $W$ is the bandwidth,  $\text{SINR}_{ji} \triangleq \frac{P_j G_{ji}}{\sum_{k\neq j}P_k G_{ki} +N_0}$, and the index $ji$ is reversed on purpose to denote directivity,  $P_j$ denotes the transmit power, $G_{ji}$ the path loss,  $N_0$ the strength of thermal noise. This model has been extensively used in the literature, cf.~\cite{deVec,Spyro,Liako}. Here we  use it only as an example;  our framework only requires that  $R_{ij}$ are independent of the association variables.
At this point, it might seem reasonable to connect each device to the RRH that provides the largest $R_{ij}$--known as the \emph{maxSINR rule}--however, this results in poor performance as we explain next.

\subsection{RRH Load}

We suppose that a device $i$ requires to download traffic $\lambda_i$, which is known\footnote{This information  can either be provided by the device, or forecasted by the system.} and device-dependent. 
Also, it downloads jobs with average size $1/\mu$;  extension to device- or RRH-specific job sizes is trivial.  Given a decision $\boldsymbol \pi$ we may determine \emph{the  load of   RRH $j$} as:
\[
\rho_j(\boldsymbol \pi)=\sum_i  \frac{\lambda_i}{\mu R_{ij}}\pi_{ij}.
\]

Connecting devices to RRHs with highest $R_{ij}$ has the beneficial effect that the load contribution of each device ${\lambda_i}/{\mu R_{ij}}$ is minimized (since the term $R_{ij}$ is maximized). However, when devices are not uniformly distributed in the area and/or total traffic demand is imbalanced, some RRHs attract more connections and become overloaded. 
This will result into poor performance, because wireless service also depends on the competition between users at the associated RRH, and rapidly degrades as $\rho_j(\boldsymbol \pi)\uparrow 1$. 
For example, we may estimate the \emph{average job completion time} by $\frac{E[N_j]+1}{\mu R_{ij}}$, where $E[N_j]$ is the average number of jobs running at RRH $j$. Assuming the jobs are served according to the  \emph{processor sharing discipline}, for $\rho_j(\boldsymbol \pi)<1$, $E[N_j]$ is given by \cite{mor}
\[
E[N_j]=\frac{\rho_j(\boldsymbol \pi)}{1-\rho_j(\boldsymbol \pi)}.
\]
As $\rho_j(\boldsymbol \pi)\uparrow 1$, the average completion time of a job of the connected devices will become very high. 

Summing up the associations of device $i$, we may  quantify its average completion time   by
$\sum_{j}\frac{\pi_{ij}}{ \mu R_{ij}(1-\rho_j(\boldsymbol \pi))}$. 
Ultimately, we are interested in minimizing  the average completion time over   all devices, which leads to   
\emph{device association problem}:\vspace{-0.1in}
\begin{align}
&\min_{ \pi_{ij}\geq 0}\sum_{ij}\frac{\pi_{ij}
}{ \mu R_{ij}(1-\rho_j(\boldsymbol \pi))}\label{eq:dap}\\
\text{s.t.} & \sum_j \pi_{ij}=1\quad\forall i\in \{1,2,\dots,n\},\notag \\
& \rho_j(\boldsymbol \pi)<1\quad\forall j\in \{1,2,\dots,m\}.\notag
\end{align}
This is a non-convex continuous optimization problem, and we are interested to solve it for large dimensions. 


\subsection{Related Work}

The baseline maxSNR association rule reads ``connect each device to the station with the strongest signal''. This simple association rule works well when the mobile traffic load is low, or symmetrically scattered around stations, but otherwise it can lead to significant load imbalances. An improved rule that captures interference (but not traffic fluctuations) is the maxSINR rule ``connect each device to the station with the highest SINR'', \cite{UACRAN2}.

To improve the performance over the above rules, the problem of device association has been studied in its integral form, where each device can be associated to a single station, i.e. the association variable $\pi_{ij}\in\{0,1\}$, cf.~\cite{boi,andrews13,athanasiou14,UACRAN3}. However, the combinatorial nature of such formulations makes them inefficient for large-scale instances. 
Instead, \cite{deVec} allowed $\pi_j(x)\in [0,1]$ for all points $x$ in the plane  and proved that the optimal solution is integral almost everywhere (except boundary points). In this paper, we directly define the association variables to  take values $\pi_{ij}\in [0,1]$ with the understanding that fractional solutions correspond to multi-station coverage like CoMP \cite{CoMP}. We observe that our optimal solutions are also ``sparse'', meaning that although variables are allowed to take values in $[0,1]$, at optimality most of them will be integral (0 or 1). 

In the context of continuous device association, past work has considered the optimization of $\alpha$-optimal functions \cite{deVec,Spyro} or general convex functions \cite{Liako}. None of the past approaches can be used to solve \eqref{eq:dap}, which is non-convex. In fact,  to the best of our knowledge, solving a large non-convex problem like \eqref{eq:dap} is generally intractable. Our goal in this paper is to propose a useful methodological tool, which can be applied to very large device association problems. 

From the perspective of scalability, most past approaches do not meet the extreme requirements we consider in this paper.  An exception is perhaps the distributed algorithm of \cite{deVec}. A limitation of this prior work, however, is that it depends on the observation of the exact statistical load $\rho_j$, and it is not robust to erroneous estimates. 




\section{Optimal Transport}

\subsection{Introduction to Optimal Transport}

The concepts of OT date back to the French mathematician Gaspard Monge who studied  in 1781 the transportation of sand masses  \cite{monge}, what seems  to be one of the first  linear programming problems studied.

Although the theory of OT has been generalized  to optimization with infinite variables, here we restrict our discussion to the illustrative case of ``discrete OT''\footnote{The notion of discreteness refers to discrete probability measures, hence we have a finite number of continuous transportation variables.}, where probability  mass must be transported between two discrete  distributions $\boldsymbol p\equiv (p_1,\dots,p_m)$ and $\boldsymbol q\equiv (q_1,\dots,q_n)$, and the transportation cost from point $i$ to point $j$ is $C_{ij}$--quite often taken to be the Euclidean distance between the two points. 
Due to Kantorovich \cite{kantorovich}, we can describe  the transportation with a coupling $\pi_{ij}$, essentially a joint probability distribution $\boldsymbol\pi$ with marginals $\boldsymbol p$ and $\boldsymbol q$. The discrete OT problem can be written as:
\begin{align}
\min_{\pi_{ij} \geq 0}&\sum_{ij}C_{ij}\pi_{ij} \label{eq:OT}\\
\text{s.t.}& \sum_j \pi_{ij}=p_i,~~i=1,\dots,m,\label{eq:p}\\
& \sum_i \pi_{ij}=q_j,~~j=1\dots,n.\label{eq:q}
\end{align}
This is a linear program,  solvable in polynomial time w.r.t.~its size $mn$. 
Further, consider a bipartite graph connecting the points with links of weight $C_{ij}$, and connect each point $i$ of $\boldsymbol p$ to a virtual source with link  capacity $p_i$, and each point $j$  of $\boldsymbol q$ to  a virtual destination with link capacity $q_j$.
The discrete OT corresponds to finding a minimum cost s-t flow of one unit. 
Using network simplex \cite{orlin91}, we can obtain the solution in $O(E^2\log V)$,\footnote{Pseudo-polynomial algorithms are faster, but their runtime guarantee depends on the values of $C_{ij}$  \cite{orlin97,tarjan}.}
 which for $V=m+n$, $E=mn$, and $n=m$, becomes $O(n^4\log n)$, essentially quadratic to the input size $mn$. Although such solution is polynomial to the input size, our problem is of enormous dimensions, and hence the degree of the polynomial is important as well. Below we describe the regularized OT, a method to approximate OT in $O(n^2 \log n)$.

\subsection{Regularized OT}

The OT can be approximated in an efficient manner using an entropic regularization \cite{Cuturi}.  
We modify the objective of OT by subtracting the  \emph{entropy} $H(\boldsymbol\pi)=-\sum_{ij} {\pi}_{ij}(\log {\pi}_{ij}-1)$, weighted with the \emph{regularization strength coefficient} $\epsilon>0$. The regularized OT becomes:
\begin{align}
\min_{\pi_{ij}}&\sum_{ij}C_{ij}\pi_{ij}+\epsilon\sum_{ij} {\pi}_{ij}(\log {\pi}_{ij}-1) \label{eq:regOT}\\
\text{s.t.}& \sum_j \pi_{ij}=p_i,~~i=1,\dots,m,\notag\\
& \sum_i \pi_{ij}=q_j,~~j=1\dots,n.\notag
\end{align}

\begin{figure}[t!]
  \centering
    \includegraphics[width=0.25\textwidth]{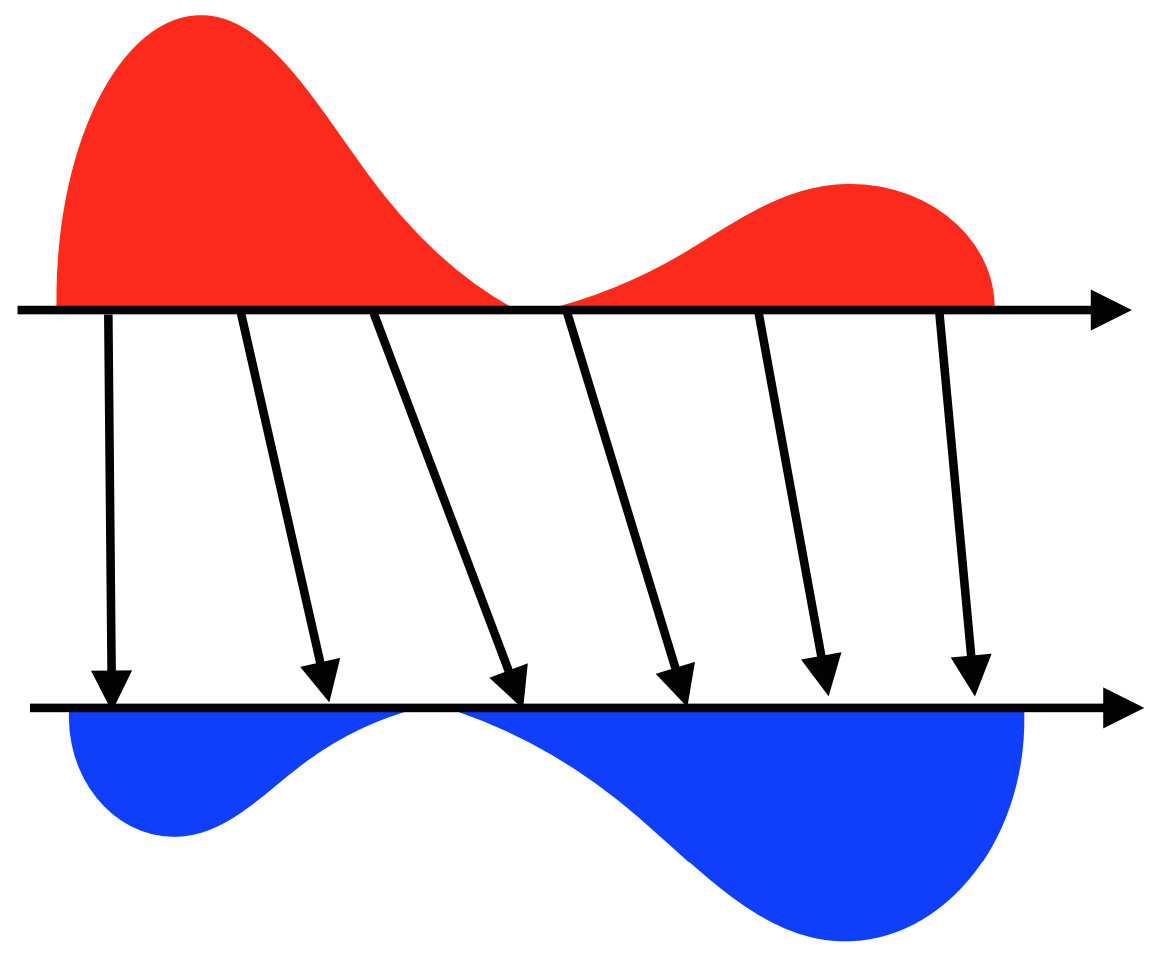}
      \caption{{\small OT studies the mass transportation cost that satisfies initial (red) and final (blue) conditions.}}
\end{figure}

Adding $-\epsilon H(\boldsymbol\pi)$ has a number of beneficial effects:
\begin{itemize}
\item The objective of the regularized OT is $1$--strongly convex, hence \eqref{eq:regOT} has  a unique optimal solution.
\item By Proposition 4.1 of \cite{COT_book}, the sequence of unique optimal solutions of \eqref{eq:regOT} for $\epsilon^{(k)}$  converges to an optimal solution to \eqref{eq:OT} as $\epsilon^{(k)}\to 0$. 
\item $H(\boldsymbol\pi)$ forces $\pi_{ij}$ to be non-negative, hence we can drop the constraint $\pi_{ij}\geq 0$.
\end{itemize}

 More importantly, the new objective admits an intuitive reformulation, which leads to a faster algorithm. 
 Consider the affine constraint sets  
 $\mathcal{C}_p=\{\boldsymbol \pi~|~\eqref{eq:p}\}$,
  and 
   $\mathcal{C}_q=\{\boldsymbol \pi~|~\eqref{eq:q}\}$,
   and let $\xi_{ij}=\exp(-{C_{ij}}/{\epsilon})$;  we can rewrite the regularized OT as: 
\[
\min_{{\boldsymbol \pi} \in \mathcal{C}_p\cup \mathcal{C}_q}  KL({\boldsymbol \pi},\boldsymbol \xi),
\]
where $KL({\boldsymbol \pi},\boldsymbol \xi)\triangleq\sum_{ij} {\pi}_{ij}(\log \frac{{\pi}_{ij}}{\xi_{ij}}-1)$ is the Kullback-Leibler (KL) divergence. Hence, the regularized OT is a KL projection of $\boldsymbol \xi$ onto the intersection of $\mathcal{C}_p$ and $\mathcal{C}_q$ \cite{BenamouEtAl}.
Since the KL divergence is the special case of the  Bregman divergence for the entropy function, our KL projections benefit from the convergence property of iterative  Bregman projections on intersections of affine constraint sets \cite{bregman67,BauschkeLewis}.
To obtain an iterative projection algorithm it is convinient to operate on the dual domain.
 Consider the Lagrangian function
\begin{align*}
L(\boldsymbol\pi,\boldsymbol \alpha, \boldsymbol \beta)&=KL({\boldsymbol \pi},\boldsymbol \xi)+\sum_i \alpha_i \left( \sum_j \pi_{ij}-p_i\right)\\
&\hspace{0.3in}+\sum_j \beta_j \left( \sum_i \pi_{ij}-q_j\right).
\end{align*}
The KKT stationarity condition requires that for each $(i,j)$ it must hold:
\[
\frac{\partial L}{\partial \pi_{ij}}=0 \quad\Leftrightarrow\quad {\pi}_{ij}=\xi_{ij}e^{-\alpha_i }e^{{-\beta_j}}.
\]
To obtain the projection with respect to $\mathcal{C}_p$ we follow the steps: (i) fix $(\beta_j)$, (ii) apply the complementary slackness condition $\sum_j \pi_{ij}=p_i$, (iii) solve for $\alpha_i$. Similarly for $\mathcal{C}_q.$

Setting $a_i\equiv e^{-\alpha_i }$ and $b_j\equiv e^{-\beta_j }$, we obtain  the Sinkhorn algorithm  \cite{sinkhorn} for regularized OT:

\noindent \rule[0.03in]{3.4in}{0.02in}

\textbf{Sinkhorn Algorithm}

\noindent \rule[0.03in]{3.4in}{0.02in}

\begin{algorithm}
 \SetKwInOut{Input}{Input}\SetKwInOut{Output}{Output}
 \Input{ $\boldsymbol C$,~$\boldsymbol p$,~$\boldsymbol q$,~$\epsilon$}
\Output{$\boldsymbol\pi$}
 initialize~~ $\boldsymbol b^{(0)}=\boldsymbol 1$, $\xi_{ij}=e^{-C_{ij}/\epsilon}$\;
 \While{accuracy}{
    $k\leftarrow k+1$ \;
    $a_i^{(k)} \leftarrow\frac{p_i}{\sum_j b_j^{(k-1)}\xi_{ij}},\quad \forall i$ \;
    $b_j^{(k)} \leftarrow\frac{q_j}{\sum_i a_i^{(k)}\xi_{ij}},\quad \forall j$ \;
%
%
 }
 $\pi_{ij}\leftarrow \xi_{ij}a_i^{(k)} b_j^{(k)},\quad \forall (i,j)$
\end{algorithm}

\noindent \rule[0.05in]{3.4in}{0.02in}
%

%
%
%
%
%
\begin{theorem}[From \cite{Altschuler17}]
Assume $m=n$, fix $\tau>0$, and choose $\gamma=\small \frac{4\log n}{\tau}$, Sinkhorn algorithm computes a $\tau$--approximate  solution of  \eqref{eq:OT} in $O(n^2\log n \tau^{-3})$ operations. 
\end{theorem}

Since the problem size is $n^2$, Sinkhorn algorithm converges almost linearly for any fixed $\tau$ (less the logarithmic term). 
This is to be contrasted with the almost quadratic convergence of network simplex $O(n^4\log n)$.
Further, Sinkhorn requires only matrix-vector multiplications, hence it admits highly efficient  GPU implementations. 
More information on computational transport can be found in \cite{COT_book}.

\subsection{Sinkhorn vs LP}

We implemented Sinkhorn in python and compared its performance to the embedded glpk  solver \cite{glpk},  known to be one of the fastest LP solvers.  Table \ref{tab:perf}  provides some indicative numerical results to highlight the  advantageous performance of Sinkhorn over  a standard LP solver, namely: (i) it is faster, (ii) scales better, and (iii) doesn't run out of memory.

In the experiments, we stopped Sinkhorn when the total absolute  residual  becomes less than $1\%$ or $0.1\%$. The runtimes provided are averages over 7 runs, computed with the \texttt{timeit} package. Notably, we can compute a 1.001-optimal transport $25\times 10k$ in half  a second.

Fig.~\ref{fig:fig1-lpa}-\ref{fig:fig1-ot1a} showcase associations obtained by Sinkhorn and LP-glpk, where we can verify the fidelity of Sinkhorn to the optimal LP solution. Fig.~\ref{fig:fig1-ot2a} showcases a very large instance that the LP solver cannot address.
\vspace{0.07in}
\begin{figure}[t!]
    \centering
    \begin{subfigure}[b]{0.15\textwidth}
        \includegraphics[width=\textwidth]{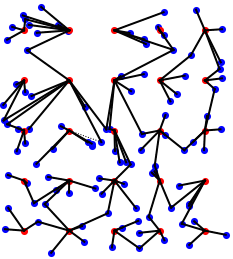}
        \caption{LP, $100 \times 25$}
        \label{fig:fig1-lpa}
    \end{subfigure} 
    \begin{subfigure}[b]{0.15\textwidth}
        \includegraphics[width=\textwidth]{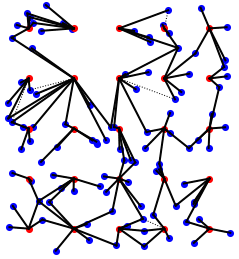}
        \caption{Reg OT, $100\times 25$}
        \label{fig:fig1-ot1a}
    \end{subfigure}
    \begin{subfigure}[b]{0.17\textwidth}
        \includegraphics[width=\textwidth]{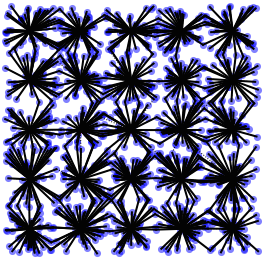}
        \caption{Reg OT, $1k\times 25$}
        \label{fig:fig1-ot2a}
    \end{subfigure}\vspace{0.05in}
     \caption{{\small Device association to  RRHs.}} 
     \label{fig:05a}\vspace{0.05in}
\end{figure}

\begin{table}[t!]
\begin{center}
  \begin{tabular}{ ccccc }
    Devices & RRHs & \emph{LP-glpk} & $\stackrel{\textit{Sinkhorn}}{\tau=1\%}$  & $\stackrel{\textit{Sinkhorn}}{\tau=0.1\%}$\\ \hline     \hline
    10          &       25 &              6&     2.27 &            5.02  \\ \hline
    50          &       25 &              88&     4.09 &            8.95  \\ \hline
    100        &       25 &              315&     8.09 &            9.7  \\ \hline
    500        &       25 &              8130&    10.1 &            31.6  \\ \hline
    1000      &       25 &              out of memory&     27 &            37.9  \\ \hline
    5000      &       25 &              out of memory&     135 &            204  \\ \hline
    10000    &       25 &              out of memory&      434 &            568 \\
    \hline
  \end{tabular}
\end{center}
\caption{{\small Runtime (msec) comparison Sinkhorn vs LP-glpk.}}
\label{tab:perf}
\end{table}

\section{OT as a Heuristic Load Balancer}

In this section we explain how one can use  OT algorithms to obtain   device association rules. 
An instance of the OT problem is defined by the triplet $(\boldsymbol C, \boldsymbol p, \boldsymbol q)$, i.e., the costs and the left-right marginals. In order to produce a load balancer based on OT, we must provide appropriate values for these parameters. We propose the following choices:
\begin{compactitem} 
\item For $\boldsymbol C$: we propose to use (i) the Euclidean distance between the location $x_i$ of device $i$ and the location $x_j$ of RRH $j$, i.e., $C_{ij}=\|x_i-x_j\|$, or (ii) the incurred load per unit traffic $C_{ij}=1/{\mu R_{ij}}$.
\item For $\boldsymbol p$: the input traffic, $p_i=\lambda_i$. 
\item For $\boldsymbol q$: we propose to use (i) equal RRH traffic $q_j=\sum_i\lambda_i/n$, (where  $n$ is the number of RRHs) or (ii) RRH traffic from the maxSINR rule, $\Lambda_j=\sum_i\pi_{ij}^{SINR}\lambda_i$.
\end{compactitem}

We provide some explanations about the choices. First, the left marginal $\boldsymbol p$ ensures that the entire traffic of each devices is split among RRHs. To choose $\boldsymbol q$ wisely, we should know the RRH traffic at optimality, however this information is often not accessible. Instead, it is easy to obtain the maxSINR rule. Last, Euclidean cost favors  nearby RRHs, while normalized load  connects devices to the RRHs with minimum  incurred load, taking into account interference and path loss. 
We have the following interesting result:
\begin{lemma}
Consider the maxSINR rule denoted with $\pi_{ij}^{SINR}$, and the incurred traffic $\Lambda_j=\sum_i\pi_{ij}^{SINR}\lambda_i$, and suppose it is  feasible, i.e., $\rho_j(\boldsymbol \pi^{SINR})<1,~\forall j$.
Also, consider a $(\boldsymbol C, \boldsymbol p, \boldsymbol q)$ instance of the OT problem, such that (i) $C_{ij}=1/{\mu R_{ij}}$, (ii) $ p_i=\lambda_i$, (iii) $q_j=\Lambda_j$, with solution $\boldsymbol x^*$.

 Then, the association $\pi_{ij}^*\doteq  x_{ij}^*/\lambda_i$ minimizes the total load $\sum_j \rho_j(\boldsymbol\pi)$.
\end{lemma}
\begin{proof}
First, note that whenever the maxSINR rule  $\pi_{ij}^{SINR}$ is feasible, it minimizes the total load. This is easy to check by  observing that the total load contribution of each device $\lambda_i\sum_j \pi_{ij}^{SINR}/\mu R_{ij}$ is minimized under the maxSINR rule. Then the total load minimization follows by summing up over devices.

The OT solution of the lemma is defined as follows: 
\begin{align}
\boldsymbol x^*\in\arg\min_{x_{ij} \geq 0}&\sum_{ij}\frac{x_{ij}}{\mu R_{ij}} \label{eq:inst}\\
\text{s.t.}& \sum_j x_{ij}=\lambda_i,~~i=1,\dots,m,\notag\\
& \sum_i x_{ij}=\Lambda_j,~~j=1\dots,n.\notag
\end{align}
Note that $\lambda_i\pi_{ij}^{SINR}$ satisfies all the constraints  of \eqref{eq:inst}, and thus it is a feasible solution of \eqref{eq:inst}. Further, since $\pi_{ij}^{SINR}$ minimizes the total load, $\lambda_i\pi_{ij}^{SINR}$ must be an optimal solution of  \eqref{eq:inst}. Therefore:
\[
\sum_{ij}\pi_{ij}^*\frac{\lambda_i}{\mu R_{ij}}=\sum_{ij}\pi_{ij}^{SINR}\frac{\lambda_i}{\mu R_{ij}}=\min_{\boldsymbol\pi}\sum_j \rho_j(\boldsymbol\pi),
\]
and the lemma follows.
\end{proof}


The lemma states that given the RRH traffic under the maxSINR rule  denoted with $\Lambda_j$,  we may use OT to retrieve the  association variables that minimize the total load, a very useful property. 

Choosing the marginals like in Lemma 1 (or with a more elaborate scheme as we show below) will work better in practice, however, the choices given at the start of the section can also be useful: they are simpler to compute and can be sufficient when the spatial traffic is uniform. To showcase this property, we experiment with  a scenario with random but uniform traffic. In a C-RAN cell with 25 RRHs we scale the number of devices from  100 to 5000. 
We compare the average completion time under 4 policies, (a) maxSINR, (b) OT with Euclidean costs and equal RRH traffic, (c) OT with load costs and equal RRH traffic, and (d) OT with load costs and RRH traffic equal to $\Lambda_j$ (same as Lemma 1). The values shown in  Fig.~\ref{fig:fig3-comp1} are the ratios of average completion time between the algorithm and the maxSINR. We observe that algorithm (d) performs the same with maxSINR (a) as predicted by the Lemma. The other two algorithms perform similarly. In particular, for a small number of devices, we see that selecting $q_j=\sum_i\lambda_i/n$ results in performance deterioration (up to 4 times worse), due to the random locations of the devices. However, as the number of devices increases, the traffic becomes more uniform and the choice $q_j=\sum_i\lambda_i/n$  becomes as good as the maxSINR solution. 

\begin{figure}
\centering
\includegraphics[width=0.8\columnwidth]{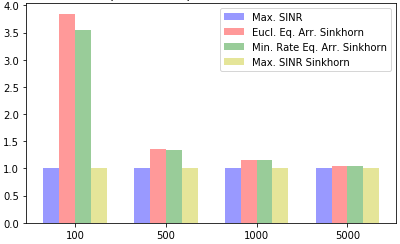}
        \caption{{\small Relative completion time (uniform traffic)}}
        \label{fig:fig3-comp1}
\end{figure}

\begin{figure*}[h!]
    \centering
    \begin{subfigure}[b]{0.22\textwidth}
        \includegraphics[width=\textwidth]{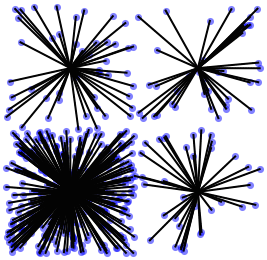}
        \caption{maxSINR}
        \label{fig:fig3-maxSINR}
    \end{subfigure}\quad
    \begin{subfigure}[b]{0.22\textwidth}
        \includegraphics[width=\textwidth]{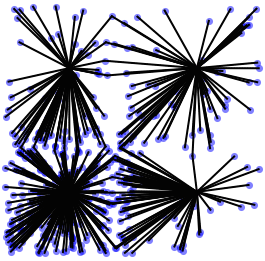}
        \caption{Adaptive  Sinkhorn}
        \label{fig:fig3-Sink}
    \end{subfigure}\quad
        \begin{subfigure}[b]{0.352\textwidth}
        \includegraphics[width=\textwidth]{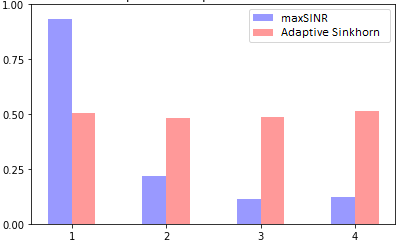}
        \caption{Load per station}
        \label{fig:fig3-comp2}
    \end{subfigure}\vspace{0.1in}
     \caption{{\small Comparison of maxSINR and Adaptive Sinkhorn (non-uniform traffic).}} 
     \label{fig:05}
\end{figure*}\vspace{0.05in}


\section{Learning RRH Load}

We now turn our attention to non-uniform traffic, as for example in figures \ref{fig:fig3-maxSINR}-\ref{fig:fig3-Sink}. In this setting, both the maxSINR rule and our OT heuristics  will overload the hot spot station. To effectively balance the load we must discover the correct traffic accumulation at RRHs which will provide the average completion time optimization.

We propose an iterative algorithm, where at iteration $k$ Sinkhorn algorithm is used with $\boldsymbol q^{(k)}$ to obtain an association which results in a specific RRH loading $\rho_j(\boldsymbol \pi^{(k)})$. According to this loading, a new marginal $\boldsymbol q^{(k+1)}$ is computed, and the process repeats until an accuracy criterion is satisfied.
More specifically, our iterative algorithm picks the RRH with the highest load (step 5) and then decreases its aggregate traffic by a fixed term $\delta$, dispersing the traffic to all other RRHs (step 6). We mention that increasing the traffic in a remote RRH will result in a large number of association changes in the Sinkhorn algorithm which will ensure that the steered traffic maintains minimum transportation cost. We provide the algorithm flow here.

\noindent \rule[0.03in]{3.4in}{0.02in}

\textbf{Adaptive Sinkhorn Association}

\noindent \rule[0.03in]{3.4in}{0.02in}

\begin{algorithm}
 \SetKwInOut{Input}{Input}\SetKwInOut{Output}{Output}
 \Input{ $C_{ij}=1/\mu R_{ij}$,~$\boldsymbol p=\boldsymbol \lambda$,~$\epsilon$}
\Output{$\boldsymbol\pi$}
 initialize~~ $q_j=\sum_i\lambda_i/n$\;
 \While{accuracy}{
 $k++$ \;
    $\boldsymbol \pi^{(k+1)}\leftarrow \text{Sinkhorn}(\boldsymbol C,\boldsymbol p,\boldsymbol q^{(k)},\epsilon)$ \;
    $j^*\in \arg\max\{\rho_j(\boldsymbol \pi^{(k+1)})\}$ \;
$q^{(k+1)}_j=\left\{\begin{array}{ll}
q^{(k)}_j-\delta & j=j^* \\
q^{(k)}_j+\delta/(n-1) & j\neq j^* 
\end{array}\right.$\;
%
%
 }
 $\boldsymbol\pi \leftarrow \boldsymbol \pi^{(k+1)} $
\end{algorithm}

\noindent \rule[0.05in]{3.4in}{0.02in}

%
%

Figures \ref{fig:fig3-maxSINR}-\ref{fig:fig3-comp2} show the results. First, comparing the two association rules we see that although the maxSINR rule associates the devices according to  interference and not traffic, our adaptive Sinkhorn algorithm considers both, and converges to an association where some cell edge devices are steered to the neighboring RRHs. This is done to alleviate the load of the bottom-left RRH. Indeed,  figure \ref{fig:fig3-comp2} shows the resulting loads of the 4 RRH. Our approach successfully equalizes the loads of the different RRH, while the maxSINR rule fails to do so. Ultimately, our scheme achieves an average completion time 6.3msec while the maxSINR rule 24msec, which corresponds to almost 4 times improvement.

\section{Conclusion}
In this paper, we studied the device association in C-RAN, and proposed an iterative algorithm to adjust the loads based on the theory of Optimal Transport. Specifically, we showed that an extension of the Sinkhorn algorithm for C-RAN systems can provide low delay associations for thousands of users in 0.5sec.  Our methodology scales to very large problem instances, and has the potential to provide great improvements over the simple baseline approach.

%

\end{document}